\documentclass[pdflatex,sn-basic,fleqn]{sn-jnl}

\usepackage{graphicx}%
\usepackage{multirow}%
\usepackage{amsmath,amssymb,amsfonts}%
\usepackage{amsthm}%
\usepackage{mathrsfs}%
\usepackage[title]{appendix}%
\usepackage{xcolor}%
\usepackage{textcomp}%
\usepackage{manyfoot}%
\usepackage{booktabs}%
\usepackage{algorithm}%
\usepackage{algpseudocode}%
\usepackage{listings}%
\usepackage{anyfontsize} 
\usepackage{mleftright}
\usepackage{mathtools}
\usepackage[capitalise, noabbrev]{cleveref}

\theoremstyle{thmstyleone}%
\newtheorem{theorem}{Theorem}
\theoremstyle{thmstyletwo}%
\newtheorem{remark}{Remark}%

\theoremstyle{thmstylethree}%
\newtheorem{definition}{Definition}%
\newtheorem{lemma}[theorem]{Lemma}
\newtheorem{fact}[theorem]{Fact}

\newcommand{\floor}[1]{\mleft\lfloor {#1} \mright\rfloor}
\newcommand{\Prob}[1]{\Pr\left(#1\right)}
\newcommand{\brc}[1]{\left\{ {#1} \right\}}
\newcommand{\set}[1]{\brc{#1}}%
\newcommand{\pth}[1]{\left({#1}\right)}%

\newcommand{\cardin}[1]{\left\lvert {#1} \right\rvert}%
\begin{document}

\title[Exact Learning of Weighted Graphs Using Composite Queries]{Exact Learning of Weighted Graphs Using Composite Queries}


\author[1]{\fnm{Michael T.} \sur{Goodrich}}\email{goodrich@uci.edu}

\author[1]{\fnm{Songyu (Alfred)} \sur{Liu}}\email{songyul4@uci.edu}

\author[1]{\fnm{Ioannis} \sur{Panageas}}\email{ipanagea@uci.edu}

\affil[1]{\orgdiv{Department of Computer Science}, \orgname{University of California, Irvine}, \orgaddress{\city{Irvine}, \state{CA}, \country{USA}}}

\maketitle

\begin{abstract}

In this paper, we study the \emph{exact learning} problem
for weighted graphs, where
we are given the vertex set, $V$, of a weighted graph, $G=(V,E,w)$,
but we are not given $E$. 
The problem, which is also known as \emph{graph reconstruction},
is to determine all the
edges of $E$, including their weights, by asking 
queries about $G$ from an oracle.
As we observe, using simple shortest-path length queries is not sufficient,
in general, to learn a weighted graph. 
So we study a number of scenarios where it is possible to learn
$G$ using a subquadratic number of composite queries, which combine
two or three simple queries.
\keywords{graph theory, machine learning, exact learning, graph reconstruction}
\end{abstract}

\section{Introduction}
The problem of exactly learning the edges of a graph from queries,
which is also known as ``graph reconstruction,'' is well-studied
and well-motivated for both practical and theoretical reasons; e.g.,
see~\cite{mazzawi_optimally_2010,abrahao_trace_2013,assadi_graph_2021,%
chen_detecting_1989,fang_detecting_1990,shi_structural_2001,shi_diagnosis_1999,%
stegehuis_reconstruction_2024,%
choi_polynomial_2013,culberson_fast_1989,kim_finding_2012,%
afshar_reconstructing_2020,mathieu_simple_2021,jagadish_learning_2013,%
wang_reconstructing_2017,%
arunachaleswaran_reconstructing_2023,kannan_graph_2018,afshar_mapping_2022,%
afshar_efficient_2022,afshar_exact_2022,eppstein_bandwidth_2025}.
For instance, from a practical perspective, there are many contexts
in which we know the vertices of a graph but must issue queries
to learn its edges.
Examples include social network science, where we know the members
of a social network, but must issue queries to learn their relationships,
such as friendship, collaboration, or even romantic partners; see, e.g.,
\cite{social,social2,social3}.
Other practical examples include learning 
phylogenetic relationships through genetic tests~\cite{afshar_reconstructing_2020},
learning faults in circuit boards~\cite{chen_detecting_1989,fang_detecting_1990,shi_structural_2001}, 
learning road networks from GPS traces~\cite{afshar_efficient_2022},
and 
learning connections in a network, such as a peer-to-peer 
network or the 
Internet via routing 
queries~\cite{afshar_mapping_2022}.

Furthermore, from the perspective of theoretical computer science,
there are many interesting complexity theory
and information theory questions that arise in graph reconstruction (e.g.,
see~\cite{mazzawi_optimally_2010,%
stegehuis_reconstruction_2024,%
choi_polynomial_2013,%
afshar_reconstructing_2020,mathieu_simple_2021,%
wang_reconstructing_2017,krivelevich_reconstructing_2025,%
arunachaleswaran_reconstructing_2023,kannan_graph_2018,afshar_mapping_2022,%
afshar_efficient_2022,afshar_exact_2022}) including the following:

\begin{itemize}
\item
What types of queries are sufficient to exactly learn
all the edges in a graph and which ones are not? 
\item
What kinds of queries
can learn the edges of a graph using a subquadratic number of queries?
\item
What are the best algorithms for solving the graph reconstruction problem
for a given set of queries?
\item
Are there non-trivial lower bounds for the number 
of queries needed to solve the graph reconstruction problem?
\item
What types of graphs are amenable to efficient exact learning algorithms?
\end{itemize}

Thus, there is ample practical and theoretical motivation
for studying the graph reconstruction problem,
where the vertices of the graph are known and we can only learn information
about the edges through queries to an oracle with full knowledge 
of the graph.
The goal is to find all the edges while minimizing the number of
queries to the oracle.  

Most of the existing literature on the graph reconstruction problem focuses on
the case where the unknown graph is connected and unweighted, with
notable results including
work by Kannan, Mathieu, and Zhou~\cite{kannan_graph_2018,mathieu_simple_2021}.
In this paper, we
study the consequences of dropping these two assumptions, which leads
to some interesting avenues of algorithmic exploration.
For example, as we observe, previous types of queries that are 
sufficient to learn unweighted connected graphs turn out to be insufficient
to learn all disconnected weighted graphs.
Thus, we need new types of queries to solve the graph reconstruction problem
for disconnected weighted graphs.
Indeed, taking inspiration from how combinations of drugs can be 
used to treat diseases~\cite{drugs}, 
such as cancer and AIDS, the methods we explore in this paper
involve solving the graph reconstruction problem
through composite queries, which combine different types of queries in concert
to efficiently learn the edges and their weights in a graph.

\subsection{Queries}
Let us therefore consider different types of queries that might be 
used to exactly learn a weighted graph efficiently.
Assume that a given weighted graph, $G$, has $n$ known vertices, $m$ 
unknown edges, and let $D$ denote the
maximum degree of any vertex in the graph.

Wang and Honorio study the problem of reconstructing
weighted directed trees
using what the authors call ``additive queries,''
which return the sum of the edge weights on the directed path
between two given vertices and 0 otherwise \cite{wang_reconstructing_2017}. 
Of course, unlike in trees, there can
be multiple paths between two vertices in a general graph; hence,
let us define the following analogous query for a possibly
disconnected weighted graph:

\begin{itemize}
\item
\textbf{Distance query}:
define the query, $q_d$,
to return the sum of the edge weights on a shortest weighted path between two
given vertices and $+\infty$ if there is no path. 
\end{itemize}

Note that we use
$+\infty$ instead of 0 so that the distance returned by the query
satisfies the triangle inequality.  
Kannan {\it et al.}~\cite{kannan_graph_2018}
show that a connected unweighted graph can be learned using
an oracle that returns 
an entire shortest path between a given pair of vertices, or an oracle that returns the distance of a shortest path between them, where distance
is measured by the number of edges, which is also known as the ``hop count.'' 
Unfortunately, as \cref{fig:distance} shows,
the weighted shortest path may not correspond to the unweighted shortest path;
hence, a solution to the weighted case
cannot easily be derived from existing solutions to the unweighted case.
Moreover, the distance query, $q_d$, only returns a single number, not
an entire path in the graph.

\begin{figure}[htb]
    \centering
    \begin{minipage}{0.48\textwidth}
        \centering
        \includegraphics[width=\textwidth]{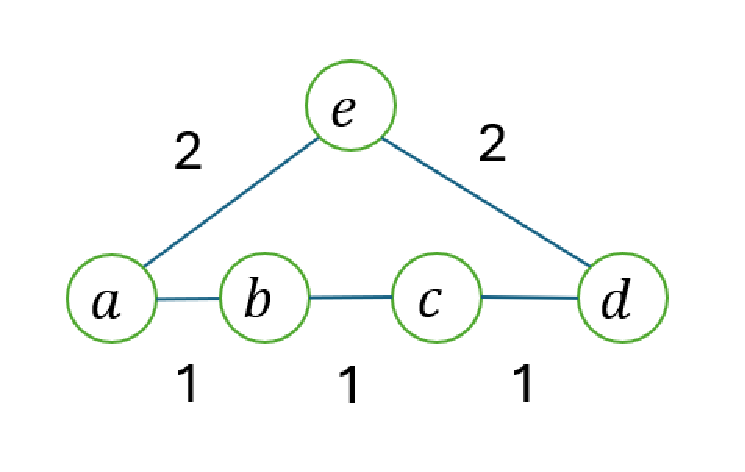}
        \caption{The weighted shortest path and the unweighted shortest path from $a$ to $d$ are different. In a weighted graph, the bottom path is the shortest. However, if we ignore the weights and treat it as an unweighted graph, then the top path will be the shortest with only two hops.}
        \label{fig:distance}
    \end{minipage}\hfill
    \begin{minipage}{0.48\textwidth}
        \centering
        \includegraphics[width=\textwidth]{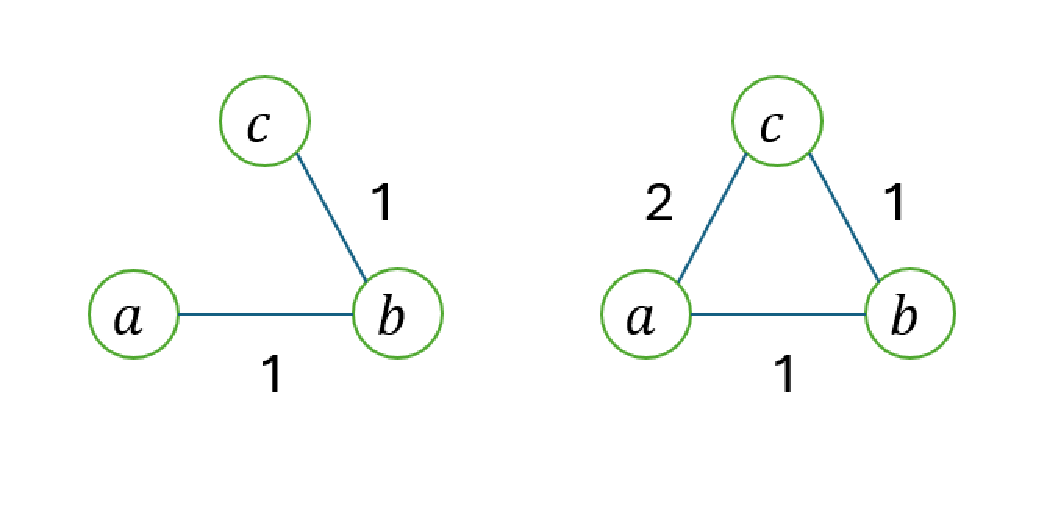}
        \caption{Using query $q_d$ alone, the two graphs will lead to the same query results for all pairs of vertices. The existence of the edge $ac$ cannot be confirmed or denied. $ac$ is called a transitive edge (\cref{def:transitive}) in the graph on the right.}
        \label{fig:transitive}
    \end{minipage}
\end{figure}
Indeed,
there are 
weighted graphs that cannot be reconstructed using $q_d$ queries alone. 
See Figure~\ref{fig:transitive}.
Related to this issue, 
the notion of transitive edges was introduced to describe similar issues 
by Wang and Honorio~\cite{wang_reconstructing_2017} 
and we redefine it here in our context.
\begin{definition}
    \label{def:transitive}
Let $G = (V,E,w)$ be an  undirected and weighted graph where
 $w(u,v)$ returns the weight of an edge, $uv$,
and 0 if there is no edge. 
An edge $uv$ is \emph{transitive} 
if a path exists between $u$ and $v$ whose weighted length $l \leq w(u, v)$.
\end{definition} 
Transitive edges cannot be recovered by $q_d$ alone. It might be possible to exclude such edges by imposing specific constraints on the graph. However, we want our algorithm to work on a more general class of graphs, so we need an additional 
type of query. In particular, we introduce the following additional
type of query:

\begin{itemize}
\item 
\textbf{Edge-weight query}:
the query,
$q_w$, takes two vertices, $u$ and $v$, and returns $ w(u, v)$. 
\end{itemize}

Using $q_w$ alone, it is at least possible
to reconstruct any weighted graph, but it is necessary to 
query every pair of vertices, leading to an $\Theta(n^2)$ brute-force
solution. We refer to this algorithm as ``EXHAUSTIVE-QUERY.''
To accommodate disconnected graphs and obtain better query complexities, therefore,
we define yet another type of query and we consider layering the input graph
based on weight thresholds. In network mapping, tools such as \verb|traceroute| provide delay information, which can be seen as distances between vertices \cite{afshar_mapping_2022, kannan_graph_2018}. Weights above a certain threshold can indicate high delay in the network.

Dividing a weighted graph into layers based on the weights of the edges revealed interesting within-layer and across-layer connections between topology and edge weights in real-world datasets \cite{bu_interplay_2023}. This motivates us to layer the input graph. Following the definition in \cite{bu_interplay_2023}, layer $W_{\textrm{thr}}$ of a graph $G$ is defined as $G[w \geq W_{\textrm{thr}}]$, the subgraph of $G$ including all vertices of $G$, but only edges with weight $w \geq W_{\textrm{thr}}$. $W_{\textrm{thr}}$ is a weight threshold that will change in our algorithm. This subgraph may be disconnected. If $w_1 < w_2, G[w \geq w_1] \supseteq G[w \geq w_2]$. 
All the queries can take an additional parameter, $ W_{\textrm{thr}}$, and give us information within this subgraph. This ability to query within a subgraph is essential because the distances between vertices in this subgraph may differ significantly from their distances in the original graph. There is no straightforward relationship between these two distances under our general assumptions on edge weights that we will explain in the next section. 

Let $q_d(u, v, W_{\textrm{thr}})$ return the sum of the edge weights on a shortest weighted path between vertices $u$ and $v$ in $G[w \geq W_{\textrm{thr}}]$, and $+\infty$ if there is no such path. 
Similarly, let $q_w(u, v, W_{\textrm{thr}})$ return $ w(u, v)$ in $G[w \geq W_{\textrm{thr}}]$ and 0 when there is no edge. Let EXHAUSTIVE-QUERY($V, W_{\textrm{thr}}$) call $q_w(u, v, W_{\textrm{thr}})$ for every pair of vertices in $V$. 
To work with disconnected graphs, we need one more type of query,
which can detect connected components efficiently. 
Relying solely on $q_d$ for this purpose would result in prohibitively high query complexity. Suppose a graph has $n$ vertices and $k$ connected components.\cite{liu_tight_2022} showed a lower bound of $kn - \binom{k+1}{2}$ queries and $k$ can be as large as $\Theta(n)$ in our case.
This necessitates another type of query.
Liu and Mukherjee \cite{liu_tight_2022} proposed the following oracle to learn $k$ connected components using $\Theta(n \log k)$ queries: 
for any vertex $u$ and set $S \subseteq V$ not containing $u$, 
\begin{gather*}
    \alpha_m(u, S) = 
\begin{cases} 
1 & \text{if for some $v \in S$, $u$ and $v$ belong to the same component}, \\
0 & \text{otherwise.}
\end{cases}
\end{gather*}
Incorporating the weight threshold, we have:
\begin{itemize}
\item
\textbf{Component query}:
 let $q_c(u, S,  W_{\textrm{thr}})$  return $\alpha_m(u, S)$ in  $G[w \geq W_{\textrm{thr}}]$.
\end{itemize}

Our composite query $q$ therefore, is a mixture of $q_w(u, v, W_{\textrm{thr}}), q_d(u, v, W_{\textrm{thr}})$, and $q_c(u, S,  W_{\textrm{thr}})$. 
Every time we perform a query, we can choose the query type that is most important. 
In addition, we assume that edges in  the original input graph have weight at least 1, and allow composite queries that
fix $W_{\textrm{thr}} = 1$, in which case these queries return 
the corresponding quantities in the original input graph.

\subsection{Related Work}
As mentioned above,
existing literature on the graph reconstruction problem
tends to focus on connected unweighted graphs 
and the following query oracles defined for such graphs,
with distance determined by hop count,
e.g., see~\cite{kannan_graph_2018,afshar_mapping_2022}:

\begin{itemize}
\item
An unweighted \emph{distance} oracle in an unweighted graph
returns the hop-count distance between a given pair of vertices.
\item
A \emph{shortest-path} oracle returns an ordered list of 
the vertices of a shortest path between a given pair of vertices. 
\item
A \emph{$\mathbf{k}$th-hop} oracle returns the $k$th vertex on 
a shortest path between a given pair of vertices.
\end{itemize}

For example,
Kannan, Mathieu, and Zhou~\cite{kannan_graph_2018} present an
algorithm for reconstructing an unweighted, connected, bounded-degree
graph using 
$\tilde{O}(n^{3/2})$ distance queries,\footnote{As is common,
  we use $\tilde{O}(*)$ to denote asymptotic bounds that ignore
  polylogarithmic factors.}
where one randomly selects centers and then partitions the 
graph into slightly overlapping Voronoi-cell 
subgraphs based on the chosen centers until
the graph clusters were all small enough to perform an exhaustive query
on each cluster. 
Afshar {\it et al.}~\cite{afshar_efficient_2022} proposed a similar 
approach with different Voronoi cell sizes and a 
query complexity that depended on the dual graph connecting neighboring cells.
Mathieu and Zhou~\cite{mathieu_simple_2021} focused on random regular graphs,
giving an algorithm that uses only $\tilde{O}(n)$ distance queries 
in this special case.

If we view the graph as an unknown non-negative $\binom{n}{2}$ dimensional 
vector, $x_G$, with $\operatorname{supp}\left(x_G\right)$ denoting the non-zero coordinates, the following queries can be defined \cite{assadi_graph_2021}:

\begin{itemize}
    \item \emph{Linear queries}: Given any non-negative $\binom{n}{2}$ dimension vector $a_G$, what is $a_G \cdot x_G$?
    \item \emph{OR queries}: Given any subset $S$ of the $\binom{n}{2}$ dimensions, is $\operatorname{supp}\left(x_G\right) \cap S$ empty?
\end{itemize}

Linear queries are more general than cross-additive queries, which in weighted graphs return the sum of the weights of the edges crossing between two given sets of vertices,
and Choi~\cite{choi_polynomial_2013} presents one of the 
few existing results for reconstructing graphs with real weights
using cross-additive queries.

A closely related, and more powerful query, is the general
\emph{additive} query. 
For any given set of vertices, the additive query returns the sum of the weights of the edges with both endpoints in the set.\footnote{
  Note that this is not the same as the additive query defined by
  Wang and Honorio in the context of weighted 
  trees~\cite{wang_reconstructing_2017}.}
Using this additive query,
Mazzawi~\cite{mazzawi_optimally_2010} gives a reconstruction algorithm that matches the information-theoretic lower bound for graphs with positive integer weights. 
Alternatively,
when the weights are bounded real numbers that are all positive (or negative), 
Kim~\cite{kim_finding_2012} reduces the graph reconstruction
problem using additive queries to a coin-weighing problem.

%

\subsection{Our Results}
In this paper, we study weighted graph reconstruction problems for graphs that are not necessarily connected and for connected graphs, and in both cases weights
are real numbers in the range, $[1,W_{\textrm{max}}]$.
\cref{tab:results} shows our results. 
The first row where the weighted graph need not be connected is our main contribution. The second row is a natural generalization of the result presented in ~\cite{kannan_graph_2018} and we provide the detailed analysis for this case in \cref{sec:without}. The main body of this paper will focus on the first case.
The idea of our approach is to divide the graph into layers and 
use a Voronoi-cell approach
to reconstruct each connected component in each layer. 
This way we will discover edges with weight in consecutive intervals,
$[1, 2), [2, 4), [4, 8)$, and so on. Although an unweighted tree reconstruction algorithm based on layers was proposed in \cite{bastide_tight_2025}, it is not comparable to our approach. In \cite{bastide_tight_2025}, layers were defined to partition vertices based on their distances to the root of the tree. This is more natural in their context because trees admit a balanced vertex separator of size 1 \cite{bastide_tight_2025}. However, it is more convenient to define layers based on edges in our case.

We assume the input graph has $m$ edges and maximum degree $D \gg 1$. We further assume that $m / D = \omega(1)$\footnote{$\omega(1)$ denotes a function that tends to infinity as the underlying parameter $n$ tends to infinity.}. 
This condition of $m / D = \omega(1)$ is relatively mild. For instance, if $D = o(n)$ and the original input graph is connected, then $m \geq n - 1$ and the condition will be satisfied. For reconstructing graphs of unbounded degree using distance queries, Kannan \textit{et al.} established an $\Omega(n^2)$ lower bound, demonstrating that bounded degree $D = o(n)$ is indeed a necessary assumption ~\cite{kannan_graph_2018}.
We can further simplify the results when $D=O(\operatorname{polylog}(n))$, an assumption made by Kannan \textit{et al.}~\cite{kannan_graph_2018}. 
However, we highlight the dependency on $D$ in our query complexity results. 
If $W_{\textrm{max}} = 1$, then 
the input graph is effecitvely unweighted, which implies 
a query complexity 
of $\tilde{O} ( D ^ {3  } n ^ {3/2})$ when the input graph is connected,
matching the complexity of the algorithm by Kannan {\it et al}.

\begin{table}[hbt]
    \caption{Our results for weighted graphs with maximum degree $D$. $\alpha$ is a positive constant to be defined in the next section.}
    \centering
    \begin{tabular}{ccc}
    \hline
   \bf Connected? &\bf Queries& \bf Query complexity\\
     \hline         
    No &$q_w(u, v, W_{\textrm{thr}}), q_d(u, v, W_{\textrm{thr}})$, $q_c(u, S,  W_{\textrm{thr}})$
      &  $\tilde{O} \pth{\pth{1 +  \frac{1}{\alpha} \log D} D ^ 3 n ^ {3/2}}$\\
      
    Yes&$q_w(u, v, 1), q_d(u, v, 1)$ &  $\tilde{O} ( D ^ {3   W_{\textrm{max}}} n ^ {3/2})$\\
     \hline
    \end{tabular}

    \label{tab:results}
\end{table}

Our algorithms introduce a number of interesting techniques for
solving the graph reconstruction problem for
weighted and possibly disconnected graphs.
For example, the approach of using queries with weight thresholds allows
us to design efficient query algorithms even with edge weights that can 
vary significantly.
Typically, when adapting algorithms from unweighted graphs to weighted graphs, the runtime becomes dependent on the aspect ratio, defined as the maximum weight divided by the minimum weight \cite{bernstein_are_2024}. Therefore, it is desirable to have a graph with a small aspect ratio. Transforming a graph into another one with a smaller aspect ratio while preserving its shortest path structure is extremely difficult \cite{bernstein_are_2024}. In contrast, our approach goes beyond adapting and makes it possible to get rid of this dependency.
This stems from our probabilistic model where edge weights are randomly drawn from a distribution. Such a model may represent scenarios like road networks, where vertices correspond to intersections, edges to roads, and edge weights to traffic conditions. While the graph structure remains fixed, the traffic follows a probability distribution. This probabilistic model is a fundamental distinction between our work and that of~\cite{kannan_graph_2018}. Building upon a specific edge weight distribution, we are able to eliminate the $\log_2{W_{\textrm{max}}}$ factor that appears in our algorithm, which is a key contribution of this paper.
Furthermore, existing literature only uses one type of query to reconstruct graphs. A single type of query may have limitations but at the same time may excel at certain tasks. The constituent parts of our composite query either fail to reconstruct the graph or do so inefficiently. However, we can design a powerful algorithm by combining them.

\section{Preliminaries}
We largely follow the notation introduced by
Kannan {\it et al.}~\cite{kannan_graph_2018}, for the sake of consistency. 
Let $G = (V,E,w)$ be an  undirected and weighted graph, where $V$ is the vertex set, $E$ is the edge set, and $w$ is the weight function. When the context is clear, we also use $w$ to denote the weight of a single edge. 
Let $d(u,v)$ denote the weighted shortest path distance of the underlying graph. The underlying graph can change since we work with different layers of the graph and different connected components. In most cases, the distance is the one returned by the query $q_d$ and should be clear from the context.
For a subset of vertices $S \subseteq V$ and a vertex $v \in V$, define $d(S, v)= \min_{s \in S} d(s, v)$. Define $\delta$ as the hop distance, i.e. the unweighted shortest path distance.
For $v \in V$, define the neighborhood of $v$ as $N(v) = \{ u \in V : \delta(u, v) \leq 1 \}$, and define the neighborhood of $v$ within 2 hops as $N_2(v) = \{ u \in V : \delta(u, v) \leq 2 \}$. 

For a pair of distinct vertices $u, v \in V$, we say that $uv$ is an edge of $G$ if and only if $uv \in E$.
For a subset of vertices $S \subseteq V$, let $G[S]$ be the subgraph induced by $S$. A closely related notation is  $G[w \geq W_{\textrm{thr}}]$, the edge induced subgraph of $G$ including only edges with weight $w \geq W_{\textrm{thr}}$ where $W_{\textrm{thr}}$ is a weight threshold.
When we pass a vertex $s \in V$ and a subset $T \subseteq V$ to a query, we perform the query for every pair in their Cartesian product. When we pass two subsets of vertices, we also perform the query for every pair in their Cartesian product. 

Recall that edges have real weights of at least 1. 
We explore two theoretical models regarding the weights, a general one and a specific one. When the input is a weighted graph or tree, existing literature assumes that the weights are 
fixed numbers~\cite{choi_polynomial_2013,wang_reconstructing_2017}.
In contrast, we study a model where the structure of the input graph is fixed and the edge weights are sampled from a continuous distribution whose support is $[1, +\infty)$. We use $F(w)$ to denote its cumulative distribution function (CDF).
Suppose the graph has $m$ edges and we have $m$ independent and identically distributed (i.i.d.) samples from this distribution. The maximum weight of the graph $W_{\textrm{max}}$ is the largest order statistic and its CDF is $F(w) ^ m$. We assume that the observed value of this random variable in the input graph is known to the algorithms and we also denote it by  $W_{\textrm{max}}$. The distinction will be clear from the context.

In \cref{sec:without}, we make no further assumptions about the distribution of edge weights, only 
that they fall within the range $[1, W_{\textrm{max}}]$. \cref{sec:with} maintains this 
range constraint but introduces a specific distribution for the weights. 
Experiments on real-world graphs suggest that edge weights follow a power law distribution (Pareto distribution) \cite{kumar_retrieving_2020}. Specifically, we use a Pareto distribution with lower threshold 1 and parameter $\alpha > 0$. Therefore, the edge weights are i.i.d. random variables with CDF
\begin{gather*}
    F(w) = 1 - w ^ {- \alpha}, \text{~~~for~} w\in [1, +\infty).
\end{gather*}
If we want to make sure the variance exists, we can assume $\alpha > 2$. However, as long as it is a positive constant, our results will hold.

\section{Queries with a Threshold} \label{sec:with}
In this section, we design a layer-by-layer reconstruction (LBL-R) algorithm (\cref{alg:layers}) that can reconstruct weighted graphs that are not necessarily connected using composite queries.
We then prove its correctness and show that its query complexity is $\tilde{O} \pth{\pth{1 +  \frac{1}{\alpha} \log D} D ^ 3 n ^ {3/2}}$.

\begin{algorithm}
\caption{Layer-by-layer reconstruction (LBL-R)}\label{alg:layers}
\begin{algorithmic}[1]
\Function{LBL-R}{$V$}
\State $ E \gets \set{}$  
\For{$j \gets 0, 1, 2, \dots, \floor{\log_2{W_ {max}}}$}
\State $W_{\textrm{thr}} \gets 2 ^ j$ 
\State $C \gets $ FIND-CONNECTED-COMPONENTS($V, W_{\textrm{thr}} $)

\For{$c \in C$} 
\If{${\cardin{c}} \leq n ^ {1/4}$} \Comment{A small connected component.}
\State $E_{jc} \gets$ EXHAUSTIVE-QUERY($c,  W_{\textrm{thr}}$)     

\Else
\State $E_{jc} \gets$ RECONSTRUCT($c, W_{\textrm{thr}} $)
\EndIf
\State $E \gets E \cup E_{jc}$
\EndFor

\If{$\max_{c \in C}{\cardin{c}} \leq n ^ {1/4}$} \Comment{All the connected components are small.}
\State break
\EndIf

\EndFor
\State \Return $E$
\EndFunction
\end{algorithmic}
\end{algorithm}

\subsection{Description of LBL-R}
First, we provide a high-level description of our algorithm. This algorithm takes as input the set $V$, which is all the vertices in the given graph, and recovers all the edges of the graph. It works iteratively.
During iteration $j$, we set $W_{\textrm{thr}} = 2 ^ j$ and find the connected components of $G[w \geq W_{\textrm{thr}}]$. 
If a connected component is small, then we can afford an exhaustive search. This will discover all of its edges with weight $w \geq W_{\textrm{thr}}$.
If a connected component is large, we will call RECONSTRUCT, which has a lower query complexity than an exhaustive search, but is not guaranteed to find the same set of edges. The missing edges will be found in subsequent iterations.
 If all the connected components are small, we will exhaustively search all of them and this is the time to break out of the outer loop.
\begin{remark}
    In fact, once EXHAUSTIVE-QUERY is called on a small connected component in iteration $j$, we do not need additional queries involving its vertices. It is possible to exclude these vertices after this iteration. The asymptotic query complexity will be the same, so we do not include this change for simplicity. 
\end{remark}

\begin{figure}[htb]
    \centering
    \begin{minipage}{0.45\textwidth}
        \centering
        \includegraphics[width=\textwidth]{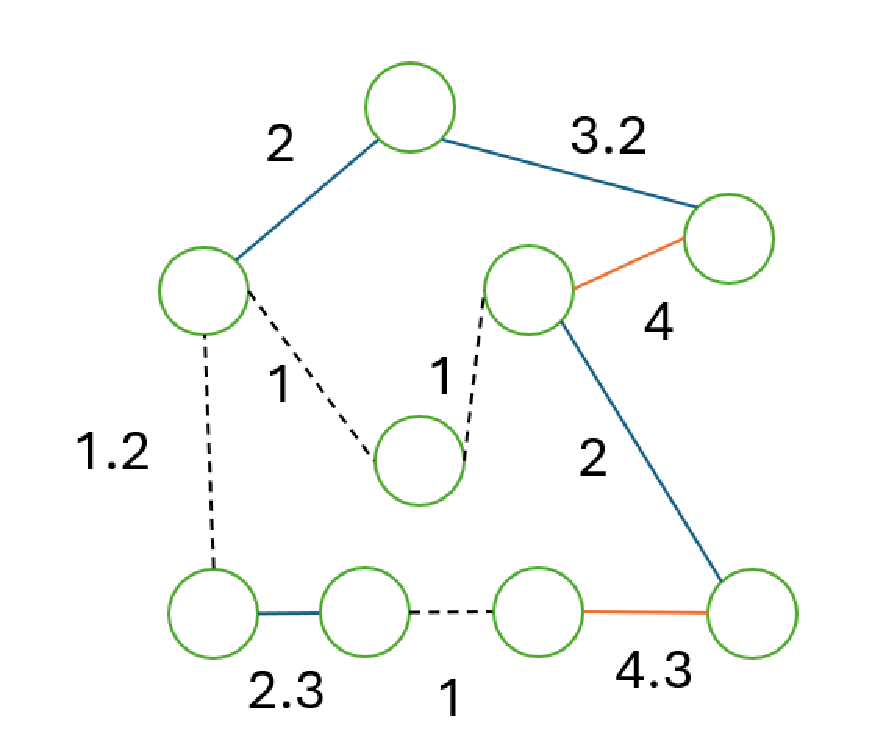}
    \end{minipage}\hfill
    \begin{minipage}{0.5\textwidth}
        \caption{Iteration 1 of LBL-R (\cref{alg:layers}). $W_{\textrm{thr}} = 2$. Edges with weight $w < W_{\textrm{thr}}$  have been discovered and are represented by dashed lines.  $G[w \geq W_{\textrm{thr}}]$ has three connected components, one of which only has one vertex. 
    Edges with weight $w  \in [W_{\textrm{thr}}, 2 W_{\textrm{thr}})$ are guaranteed to be discovered and are represented by solid blue lines. Edges with weight $w \geq 2 W_{\textrm{thr}}$ will be discovered in future iterations and are represented by solid orange lines.}
        \label{fig:iteration}
    \end{minipage}
\end{figure}

\begin{algorithm}
\caption{Finding connected components} \label{alg:components}
\begin{algorithmic}[1]
\Function{FIND-CONNECTED-COMPONENTS}{$V, W_{\textrm{thr}}$}
\State Arbitrarily order $V = \set{v_1, v_2, \dots}$
    \State Initialize $C_1 = \{v_1\}$, $k = 1$
    \For{$i = 2$ to $\cardin{V}$}
        \If{$q_c(v_i, C_1 \cup \cdots \cup C_k, W_{\textrm{thr}}) = 0$}
            \State $k \gets k + 1$
            \State Add $v_i$ to $C_k$
        \Else
            \State Find $j$ such that $q_c(v_i, C_j, W_{\textrm{thr}}) = 1$ via binary search among $\{C_1, \ldots, C_k\}$
            \State Add $v_i$ to its corresponding $C_j$
        \EndIf
    \EndFor
    \State \Return $\{C_1, \ldots, C_k\}$
\EndFunction
\end{algorithmic}
\end{algorithm}

   \begin{algorithm}
\caption{Finding $N_2(a)$} \label{alg:neighbors}
\begin{algorithmic}[1]
\Function{FIND-NEIGHBORS}{$V, a,  W_{\textrm{thr}}$}
    \State $N(a) \gets \{v \in V : q_w(a, v, W_{\textrm{thr}}) \neq 0\}$ \Comment{Find all neighbors of $a$.}
    \State $N_2(a) \gets N(a)$ 
    \For{$v \in N(a)$} \Comment{For each neighbor of $a$.  At most $D$ neighbors.}
        \State $N(v) \gets \{u \in V : q_w(u, v, W_{\textrm{thr}}) \neq 0\}$ \Comment{Find all neighbors of $v$.}
        \State $N_2(a) \gets N_2(a) \cup N(v)$ \Comment{Add these to $N_2(a)$.}
    \EndFor
 
    \State \Return $N_2(a)  $
\EndFunction
\end{algorithmic}
\end{algorithm}

\begin{algorithm}
\caption{Selecting centers by estimation}    \label{alg:centers}
\begin{algorithmic}[1]
\Function{ESTIMATED-CENTERS}{$V, s,  W_{\textrm{thr}}$}
\State $n \gets \cardin{V}$
    \State $A \gets \emptyset$, $W \gets V$
    \State $T \gets K \cdot s \cdot \log n \cdot \log \log n$
    \While{$W \neq \emptyset$}
        \State $A' \gets \text{SAMPLE}(W, s)$
        \State $q_d(A', V,  W_{\textrm{thr}})$ \Comment{ Every pair in $A' \times V$ is queried.}
        \State $A \gets A \cup A'$
        \For{$w \in W$}
            \State $X \gets \text{random multi-subset of } V \text{ with } T \text{ elements}$
            \State  $q_d(X, w,  W_{\textrm{thr}})$
            \State $\tilde{C}(w) \gets |\{v \in X : d(w, v) < d(A, v)\}| \cdot n/T$
        \EndFor
        \State $W \gets \{w \in W : \tilde{C}(w) \geq 5n/s\}$

    \EndWhile
    \State \Return $A$
\EndFunction
\end{algorithmic}
\end{algorithm}

\begin{algorithm}
\caption{Subroutine to reconstruct the graph once} \label{alg:reconstruct-sub}
\begin{algorithmic}[1]
\Function{RECONSTRUCT-SUB}{$V,  W_{\textrm{thr}}$}
\State $n \gets \cardin{V}, s \gets D \sqrt{n}$
    \State $A \gets \text{ESTIMATED-CENTERS}(V, s,  W_{\textrm{thr}})$  \Comment{ Every pair in $A \times V$ is queried.}
    \For{$a \in A$}
   \State $N_2(a) \gets$ \Call {FIND-NEIGHBORS}{$V, a,  W_{\textrm{thr}}$}
    
        \State  $q_d(N_2(a), V,  W_{\textrm{thr}})$ \Comment{ Every pair in $N_2(a) \times V$ is queried.}
        \For{$b \in N_2(a)$}
            \State $C(b) \gets \{v \in V : d(b, v) < d(A, v)\}$
        \EndFor
            \State $D_a \gets \bigcup \{C(b) : b \in N_2(a)\} \cup N_2(a)$
            \State $E_a \gets \text{EXHAUSTIVE-QUERY}(D_a,  W_{\textrm{thr}})$

    \EndFor
    \State \Return $\bigcup_a E_a$
\EndFunction
\end{algorithmic}
\end{algorithm}

\begin{algorithm}
\caption{Reconstructing the graph} \label{alg:reconstruct} 
    \begin{algorithmic}[1]
    \Function{RECONSTRUCT}{$V,  W_{\textrm{thr}}$}
\State $n \gets \cardin{V}$
\State Set query limit $Q = O(D^3 \cdot n^{3/2} \cdot \log^2 n \cdot \log \log n)$
\State $E \gets$ \Call{RECONSTRUCT-SUB}{$V,  W_{\textrm{thr}}$}
\If {RECONSTRUCT-SUB does not terminate after $Q$ queries}
    \State Stop RECONSTRUCT-SUB
    \State Re-execute this algorithm from the beginning
\EndIf
\State \Return $E$
\EndFunction

\end{algorithmic}
\end{algorithm}

Next, we explain the subroutines used in our main algorithm. SAMPLE$(W, s)$ receives a set of vertices $W$ and a number $s$ and returns a random subset of $W$ by selecting each element independently with probability $s/|W|$ \cite{kannan_graph_2018}. If $|W| \leq s$, then SAMPLE$(W, s)$ returns the set $W$.  As a direct result of our component query, FIND-CONNECTED-COMPONENTS  finds connected components of $G[w \geq W_{\textrm{thr}}]$ with query complexity $\tilde{O}(n)$ \cite{liu_tight_2022}.
The other subroutine RECONSTRUCT can be treated as a black box algorithm for our purpose. RECONSTRUCT and its subroutines are adaptations of the algorithms presented in \cite{kannan_graph_2018}, using our specific queries. While our LBL-R algorithm builds upon this subroutine used for the unweighted, connected case in \cite{kannan_graph_2018}, we have introduced significant modifications to reconstruct a different class of graphs. When the parameter $W_{\textrm{thr}}$ is set to 1, our RECONSTRUCT subroutine includes the original one in \cite{kannan_graph_2018} as a special case.
Now we briefly review this original version for completeness. 

For a connected graph $G_c$ with bounded degree, it first finds a set of vertices $A$ and each vertex $a \in A$ will become the center of an extended Voronoi cell $D_a \subseteq V$. These cells are overlapping. Then an exhaustive search is performed within each cell and the union of vertex-induced subgraphs $\bigcup_{a \in A} G_c[D_a] $ will cover all the edges. Finally, because this whole procedure is randomized, its query complexity may differ. If the number of queries issued is too large, it will stop the ongoing procedure and retry.

We now consider the query complexity of RECONSTRUCT.
The query complexity proofs in \cite{kannan_graph_2018} can be directly extended to accommodate any metric between vertices of a connected graph with bounded degree. We present the following lemma without proof.

\begin{lemma}[Extension of Section 2.2.2 in \cite{kannan_graph_2018}]
For connected graphs with $n$ vertices and maximum degree $D$, the randomized algorithm RECONSTRUCT has expected query complexity $\tilde{O}(D^3 n^{3/2})$.
\end{lemma}

The query complexity result assumes that $n$ is large, so we do not apply RECONSTRUCT to all connected components. Note that the expectation is with respect to (wrt) the random choices in the algorithm and does not depend on the random weights of our input graph. Also, this lemma only states the query complexity and not the correctness. In fact, RECONSTRUCT by itself may not correctly reconstruct our input graph and a more detailed analysis is in the subsequent section.

\subsection{Correctness of LBL-R}
To show the correctness of our algorithm LBL-R, we first analyze what happens in one iteration for one connected component. 
Our analysis focuses on RECONSTRUCT-SUB, which is the core procedure repeatedly executed by the RECONSTRUCT algorithm, and the result will directly extend to RECONSTRUCT. 

\begin{lemma}
Suppose we set the threshold in the queries to $W_{\textrm{thr}}$.
 $G[w \geq W_{\textrm{thr}}]$ consists of one or more connected components. Let $G_c$ be a connected component and let $c$ be its vertices.  $G_c[w < 2 W_{\textrm{thr}}]$ is the subgraph of $G_c$ induced by edges with weight $w < 2 W_{\textrm{thr}}$.
For function call RECONSTRUCT-SUB($c, W_{\textrm{thr}} $), we have
    $G_c[w < 2 W_{\textrm{thr}}] \subseteq \bigcup_{a \in A} G_c[D_a] $.
\end{lemma}
\begin{proof}
   Let $d$ be the distance in the connected component $G_c$. For any two vertices in it, the distance between them in $G_c$ equals their distance in  $G[w \geq W_{\textrm{thr}}]$, which is what $q_d$ would return with threshold $W_{\textrm{thr}}$.

   Let $uv$ be any edge of $G_c[w < 2 W_{\textrm{thr}}]$. $ d(u, v) = w(u, v) \in [W_{\textrm{thr}}, 2 W_{\textrm{thr}})$. We want to show that $\exists a \in A$, such that $u, v \in D_a$. W.l.o.g., $d(A, u) \leq d(A, v)$.  $a$ is a vertex such that $d(a, u) = d(A, u)$. 
    If there exists a shortest weighted path from $a$ to $u$ with one edge, then $u, v \in N_2(a) \subseteq D_a$. 
    
    Otherwise, any shortest weighted path from $a$ to $u$ has at least two edges. Let $b$  be the vertex two hops away from $a$ on one such path.
$d(b, u) = d(a, u) - d(a, b) = d(A, u) - d(a, b) < d(A, u)$. By the triangle inequality, 
    \begin{gather}
        d(b, v) \leq d(b, u) + d(u, v)  \\
     = d(A, u) - d(a, b) + d(u, v)  \\
        \leq   d(A, v) - d(a, b) + d(u, v) \\
        <  d(A, v) - 2 W_{\textrm{thr}} + 2 W_{\textrm{thr}} \\
          =  d(A, v) 
    \end{gather}
    Therefore, $u, v \in C(b)$. Since $b \in N_2(a), u, v \in D_a$. 
    \end{proof}
    
We can show that RECONSTRUCT will discover edges with weights $w \in [W_{\textrm{thr}}, 2W_{\textrm{thr}})$ within a single connected component during one iteration of our main algorithm.
When the threshold in the queries is $W_{\textrm{thr}} $, the algorithm might also discover edges with weight $w \geq 2 W_{\textrm{thr}}$ but there is no guarantee. 

Now we can show the correctness of LBL-R (\cref{alg:layers}). During any iteration, for a connected component, if the EXHAUSTIVE-QUERY case is executed, we will discover its edges with weight $w \geq W_{\textrm{thr}}$, which is a superset of $[W_{\textrm{thr}}, 2 W_{\textrm{thr}})$.
Otherwise, we will discover its edges with weight $w \in [W_{\textrm{thr}}, 2 W_{\textrm{thr}})$.
The union of these intervals will cover all the possible weights when the outer loop finishes. If the break statement is triggered, EXHAUSTIVE-QUERY will discover all edges with weight $w \geq W_{\textrm{thr}}$. The correctness of LBL-R (\cref{alg:layers}) thus follows.

\begin{theorem}[Main correctness]
\label{thm:all}
    The output of LBL-R is all edges of $G$.
\end{theorem}
\begin{proof}
We claim that when starting iteration $j$, we have discovered all edges with weight $w < 2 ^ j$. So even though the queries are restricted to $G[w \geq  2 ^ j]$, we will not miss any edges because of it.
If the algorithm breaks out of the main loop during iteration $j$, then EXHAUSTIVE-QUERY will be performed for all connected components, and it will discover all edges with weight $w \geq 2 ^ j$, so together all edges will be found. If the break statement is never triggered, the outer loop has enough iterations to cover all the edges.

    The proof of the claim is by induction.  
    The base case is that when starting iteration 0, we have discovered all edges with weight $w < 2 ^ 0$. This is a vacuous truth.
    The inductive hypothesis is that when starting iteration $j$, we have discovered all edges with weight $w < 2 ^ j$. We want to show that when starting iteration $j + 1$, we will find all edges with weight $w < 2 ^ {j + 1}$.

        For any edge $e = uv$ with weight $w$ such that $2 ^ j = W_{\textrm{thr}} \leq w < 2 W_{\textrm{thr}}$, we want to show this edge will be discovered during iteration $j$. 
Since $e \in G[w \geq  2 ^ j]$, it will be in one of the connected components. EXHAUSTIVE-QUERY or RECONSTRUCT will find it by the lemma above.  
\end{proof}

\subsection{Query Complexity of LBL-R}
As previously indicated, the central challenge of our analysis lies in eliminating the $\log_2{W_{\textrm{max}}}$ factor in  LBL-R (\cref{alg:layers}). 
The lemmas and proofs that accomplish this constitute the primary technical contribution of our work and distinguish us from the prior work by \cite{kannan_graph_2018}.
We first consider the query complexity for one iteration of our main algorithm. The query complexity to find all the connected components is much smaller than the other steps since we have component queries at our disposal. Then, we compute the query complexity to process all connected components. Finally, we analyze when the break statement will be triggered. 

Let us focus on the inner for loop in LBL-R (\cref{alg:layers}). For some connected components, we will perform an exhaustive search, and for the rest, we will call  RECONSTRUCT. The sizes of these connected components are carefully chosen so that for each group of components, the query complexity has a small upper bound. 

\begin{fact}
    Let $n_1, \dots, n_m > 0, p > 1$. $ \sum_{i = 1}^m n_i ^ p \leq  (\sum_{i = 1}^m n_i) ^ p $.
\end{fact}
\begin{lemma}[Query complexity for all connected components]
\label{lem:inner}
    The inner for loop in LBL-R has expected query complexity $\tilde{O}(D ^ 3 n ^ {3/2})$. The expectation is wrt the random choices in the RECONSTRUCT algorithm and this result holds for all input graphs.
\end{lemma}
\begin{proof}
An exhaustive search will only be performed on connected components with at most $n ^ {1/4}$ vertices, requiring $O(n ^ {1/2})$ queries each. There are at most $n$ connected components, resulting in $O(n ^ {3/2})$ queries. 
Suppose the rest of the connected components have $n_1, \dots, n_{l}$ vertices. $\sum_{i = 1}^{l} n_i \leq n$.
        Calling RECONSTRUCT on these connected components has query complexity 
\begin{gather}
    \sum_{i = 1}^{l} \tilde{O}(D ^ 3 n_i ^ {3/2}) 
    = \tilde{O} (D ^ 3  \sum_{i = 1}^{l} n_i ^ {3/2}) 
    \leq \tilde{O} (D ^ 3  (\sum_{i = 1}^{l} n_i) ^ {3/2}) 
    \leq  \tilde{O} (D ^ 3 n ^ {3/2}) 
\end{gather}
Assuming $ D \gg 1$, the overall expected query complexity is $\tilde{O}(D ^ 3 n ^ {3/2})$.
\end{proof}

If we directly incorporate the outer loop, the overall query complexity will depend on $\floor{\log_2{W_ {max}}}$. This is undesirable since $W_ {max}$ can be quite large. However, we can leverage our edge weight distribution to eliminate this term and significantly improve the query complexity. Recall that throughout this section, edge weights are sampled from a Pareto distribution.
For a particular input graph, all its weights are fixed. The distribution specifies all the weight configurations that this graph may have.
The following lemma shows that under our assumptions, the maximum weight is $\Omega(D ^ {\frac{1}{\alpha}})$ asymptotically almost surely (a.a.s.).\footnote{a.a.s. means the statement holds true with probability tending to one as the number of vertices tends to infinity. } While a stronger result could be established, this lemma is sufficient for our subsequent analysis and effectively demonstrates that the maximum weight is large. Recall that $m / D = \omega(1)$, the weights are $m$ i.i.d. samples, and $W_{\textrm{max}}$ is their maximum. Since the Pareto distribution has a heavy tail, when we draw sufficiently many samples from it, the maximum will be arbitrarily larger than $D^{\frac{1}{\alpha}}$ a.a.s.

\begin{lemma}[Maximum weight]\label{lem:wmax}
    Let  $w^* = D ^ {\frac{1}{\alpha}} > 1$, then for any constant $c > 1$, $ \lim_{n \to +\infty} \Prob{W_{\textrm{max}} > c w^*} = 1$.
\end{lemma}
\begin{proof}
    \begin{gather}
        \Prob{W_{\textrm{max}} > c w^*} 
=  1 -  \Prob{W_{\textrm{max}} \leq c w^*}     
= 1 - \pth{ F\pth{c w^*}}^ m
\end{gather}
    We want to show that $q \coloneq \pth{ F\pth{c w^*}}^ m $ tends to 0. Let $K = c ^ {-\alpha} $. Since $c w^* > 1$,
    \begin{gather}
    1 > 1 - F\pth{c w^*} = (c w^*) ^ {-\alpha} = c  ^ {-\alpha} (w^*) ^ {-\alpha} = K / D > 0   \\
    q      =  \pth{ 1 - K / D}^ m
        = \pth{\pth{1 - K / D} ^ D} ^ {m / D}
        \end{gather}

If $D$ is a constant, then $\pth{1 - K / D} ^ D$ is a constant in $(0, 1)$. If $D = \omega(1)$, then by the definition of the exponential function as a product limit,  $\pth{1 - K / D} ^ D \to \exp(-K)$ is also a constant  in $(0, 1)$. 
In both cases, $\lim_{n \to +\infty} \pth{1 - K / D} ^ D  \in (0, 1)$.
Since  $ m / D = \omega(1), \pth{\pth{1 - K / D} ^ D} ^ {m / D} $ approaches 0.   Thus, $q$  approaches 0.
\end{proof}

 At any iteration $j$, the weight of an edge may or may not be below the threshold. 
 For any edge, the probability that it is in $G[w \geq W_{\textrm{thr}}]$ is $p \coloneq \Prob{w \geq W_{\textrm{thr}}} \coloneq T(W_{\textrm{thr}})$. 
 We notice the similarity between this and the binomial random graph $G(n, p)$ \cite{janson_random_2000}. In $G(n, p)$, there are $n$ vertices, and any vertex can connect to all other vertices with probability $p$. The size of connected components in $G(n, p)$ has been well studied.
 In our model, we can show an analogous result which allows us to analyze at which iteration the break statement will be triggered.
 \begin{lemma}
[Part of Theorem 5.4 in \cite{janson_random_2000}]
     In a binomial random graph $G(n, p)$, if $np = c \in (0, 1)$, then a.a.s. the largest connected component has $O(\log(n))$ vertices. 
\end{lemma}

\begin{lemma}[Largest connected component]
\label{lem:largest}
In our model, if $Dp = c \in (0, 1)$, then a.a.s. the largest connected component has $O(\log(n))$ vertices. 
\end{lemma}
\begin{proof}
The proof is similar to the proof of Theorem 5.4 in \cite{janson_random_2000}.
We reveal the component structure step by step, using the following procedure. Choose a vertex $v$, find all neighbors $v_1, \ldots, v_r$ of $v$ in $G[w \geq W_{\textrm{thr}}]$ based on the randomly generated weights, and mark $v$ as \textit{saturated}. Then, generate all vertices $\{v_{11}, \ldots, v_{1s}\}$ from $N(v_1) \setminus \{v, v_1, \ldots, v_r\}$ which are adjacent to $v_1$ in $G[w \geq W_{\textrm{thr}}]$, so $v_1$ becomes saturated, and continue this process until all vertices in the component containing $v$ are saturated. $N(v_1)$ is all neighbors of $v_1$ in the original input graph $G$.

The number $X_i$ of new vertices we add to the component in the $i$-th step, provided $m$ of its elements have already been found, has binomial distribution $\text{Bi}(n_i, p)$ where $n_i = \cardin{N(v_i)  \setminus \{v, v_1, \ldots\, v_{m - 1}\}}, p = \Prob{w \geq W_{\textrm{thr}}} = T(W_{\textrm{thr}})$.

The probability that a given vertex $v$ belongs to a component of size at least $k = k(n)$ is bounded from above by the probability that the sum of $k = k(n)$ random variables $X_i$ is at least $k - 1$. Furthermore, the conditional distribution of $X_i$ can be bounded from above by $X_i^+$, where all $X_i^+$ are independent random variables following binomial distribution $\text{Bi}(D, p)$ and the bound is in terms of stochastic domination.
We have $\sum_{i=0}^{k - 1} X_i^+ \sim \text{Bi}(kD, p)$. By additive Chernoff bounds, for $n$ large enough the probability that there is a component of size at least $k \geq 3 \log n / (1 - c)^2$ is bounded from above by

\begin{gather}
    n \Prob{  \sum_{i=0}^{k - 1} X_i^+ \geq k - 1 }     \\
    = n \Prob{  \sum_{i=0}^{k - 1} X_i^+ \geq ck + (1 - c)k - 1 }
\\
\leq n \exp \left( - \frac{((1 - c)k - 1)^2}{2(ck + ((1 - c)k - 1)/3)} \right)  \\
\leq n \exp \left( - \frac{(1 - c)^2 k}{2} \right)  \\
= o(1)
\end{gather}
Taking the complement, the probability that all components have size $O(\log(n))$ is bounded from below by  $ 1 -  o(1)$.
\end{proof}

Now we can prove our main lemma in this section. The result suggests that in LBL-R (\cref{alg:layers}) the outer loop is in fact superficial. We do not need to finish the outer loop to fully recover the underlying graph. This will allow us to get rid of the query complexity's dependency on the maximum weight and obtain an efficient algorithm. The proof outline is as follows:
$W_{\textrm{thr}}$ increases with each iteration, and since weights follow a Pareto distribution, $p = \Prob{w \geq W_{\textrm{thr}}} = T(W_{\textrm{thr}})$ will decrease. Once $Dp$ falls below 1, we can apply the lemma above. Consequently, the break statement in LBL-R (\cref{alg:layers}) will be triggered no later than this point and we can calculate when this occurs.

\begin{lemma} [Early termination]  \label{lem:break}
Under the following assumptions:
\begin{enumerate}
    \item $m / D = \omega(1)$.
    \item Edge weights are i.i.d. random variables with CDF
\begin{equation*}
    F(w) = 1 - w^{-\alpha} \quad \text{for } w \in [1, \infty), \text{where } \alpha > 0 \text{ is a constant},
\end{equation*}
and the maximum of these random variables is $W_{\textrm{max}}$.
\end{enumerate}
We have the following result:
 a.a.s. the break statement in LBL-R will terminate the outer loop at least $i$ iterations early for any natural number $i$. 
\end{lemma}
\begin{proof}

We want to find the critical value $W_{\textrm{thr}}$ such that $Dp = D T(W_{\textrm{thr}}) = 1$. The unique solution to this equation is $w^* = D ^ {\frac{1}{\alpha}}$. Since $D > 1$, we have $T(w^*) = 1 / D \in (0, 1), w^* > 1$. 
Suppose $k$ is the first iteration such that $W_{\textrm{thr}} = 2 ^ k > w^*$. This is not possible at iteration 0 so $k > 0$.  
Early termination happens at or before iteration $k$ a.a.s. 

    Without the break statement, the outer loop would terminate at iteration $j = \floor{\log_2{W_ {max}}}$. Let $k' = \log_2 w^* = \frac{1}{\alpha} \log_2 D, j'= \log_2{W_ {max}}$. Since $k \leq k' + 1, j \geq j' - 1$, the number of iterations we skip is at least $j - k \geq j' - k' - 2$ a.a.s. 
Next, we study the distribution of  $I \coloneq j' - k' - 2$.  $\forall i \in \mathbb{N}$, let $c = 2 ^ {i + 2} > 1$. 
    \begin{gather}
        \Prob{I > i} =  \Prob{2 ^ I > 2 ^ i}    
    =\Prob{W_{\textrm{max}} / 4 w^*  > 2 ^ i}  
    =\Prob{W_{\textrm{max}} > c w^*} 
    \end{gather} 
Now we can invoke \cref{lem:wmax} and complete the proof.
\end{proof}

\begin{theorem} [Query complexity] \label{thm:complexity}
    For weighted graph reconstruction using composite queries, there is a randomized algorithm with query complexity $\tilde{O} \pth{\pth{1 +  \frac{1}{\alpha} \log D} D ^ 3 n ^ {3/2}}$.
\end{theorem}
\begin{proof}
The query complexity to find connected components is dominated by the other steps of LBL-R (\cref{alg:layers}). The inner for loop has expected query complexity $\tilde{O}(D ^ 3 n ^ {3/2})$. 
The early termination happens no later than iteration $k$  a.a.s. where $k \leq k' + 1, k' = \frac{1}{\alpha} \log_2 D$ as in the proof of \cref{lem:break}.
The outer for loop will be executed at most $k' + 1$ times a.a.s. Thus, the overall expected query complexity is $\tilde{O} \pth{\pth{1 +  \frac{1}{\alpha} \log D} D ^ 3 n ^ {3/2}}$  a.a.s. The expectation is wrt the random choices in the algorithm and the  a.a.s. part is wrt the randomness of the edge weights.
\end{proof}

\section{Conclusion}
In this paper, we demonstrate that for weighted and potentially disconnected graphs with bounded degree whose weights follow a Pareto distribution, combining distance queries with edge-weight and component queries enables efficient graph reconstruction. An interesting direction for future work would be to establish lower bound results when only two of the three query types are allowed for interesting subclasses of graphs. For instance, if a graph has no transitive edges, can edge-weight queries be omitted? Similarly, is it possible to reconstruct subclasses of graphs without component queries?
As previously noted, $\Omega(n^2)$ distance queries are necessary to identify all connected components in a general graph. However, stronger assumptions on the graph structure and edge weights might lead to better lower bounds.

\bmhead{Acknowledgements}
Michael T. Goodrich was supported by NSF grant 2212129.
Ioannis Panageas was supported by NSF grant CCF-2454115.

\backmatter

\bmhead{Supplementary information}
Not applicable.





\section*{Declarations}
Not applicable.



\bibliography{outside-zotero, zotero}
\clearpage

\begin{appendices}

\section{Queries without a Threshold} \label{sec:without}
In this section, we will show that it is difficult to efficiently reconstruct a weighted graph without applying different weight thresholds in the queries.
We design a no-threshold reconstruction (NT-R) algorithm that can reconstruct connected weighted graphs using distance and edge-weight queries.
For all queries, we fix $W_{\textrm{thr}} = 1$ and always work with the original input graph so there is effectively no threshold. We will then show the correctness of our algorithm and establish that its query complexity is $\tilde{O} ( D ^ {3   W_{\textrm{max}}} n ^ {3/2})$. If we compare this and the  $\tilde{O} \pth{\pth{1 +  \frac{1}{\alpha} \log D} D ^ 3 n ^ {3/2}}$ result we showed, we can see that the query complexity now has an exponent that depends on $W_{\textrm{max}}$, and by \cref{lem:wmax}, this will significantly increase the complexity. 

\subsection{Description of NT-R}
First, we describe our algorithm. NT-R is RECONSTRUCT (\cref{alg:reconstruct}) modified. It takes as input the set $V$, which is all the vertices in the input graph, and $W_{\textrm{thr}} = 1$, and recovers all the edges of the graph. 
The changes are as follows. In RECONSTRUCT, set $Q = \tilde{O} ( D ^ {3   W_{\textrm{max}}} n ^ {3/2})$.
In its subroutine RECONSTRUCT-SUB (\cref{alg:reconstruct-sub}), we set $s$ to a different value and the exact formula is in the appendix. 
In addition, we replace $N_2(a) \gets$ \Call {FIND-NEIGHBORS}{$V, a,  W_{\textrm{thr}}$} by 
$N_2 (a) \gets \bar{B}(a, 2W_{\textrm{max}}) = \set{v \in V : d(a, v) \leq 2W_{\textrm{max}}}$, i.e., the closed ball centered at $a$ with radius $2W_{\textrm{max}}$ using metric $d$.
We call this modified subroutine of NT-R NT-RS.
Note that $\bar{B}(a, 2W_{\textrm{max}})$ can be directly calculated because when the algorithm needs it, every pair in $A \times V$ has been queried.  $\forall a \in A, \forall v \in V, d(a, v) $ is available.


\subsection{Correctness of NT-R}
The correctness is a result of the changes we made and the intuition is as follows.
$\bar{B}(a, 2W_{\textrm{max}})$ includes all the vertices two hops away from $a$ and potentially much more. When we compute the extended Voronoi cell $D_a$, it will take the union of $C(b)$ over more vertices $b$. Since each cell is larger, we will discover more edges. More precisely, we have the following lemma.

\begin{lemma}[Correctness]
\label{lem:2.4}
For NT-RS,    $G \subseteq \bigcup_{a \in A} G[D_a] $ where $G[D_a] $ is the subgraph of $G$ induced by vertices in $D_a $.
\end{lemma}

\subsection{Query Complexity of NT-R}
Without the weight thresholds, the query complexity will depend on $W_{\textrm{max}} $ and this term is large.
If we assume a specific distribution for the weights, we can further analyze the distribution of $W_{\textrm{max}} $ and how it affects the query complexity. In general, we have the following:
\begin{theorem} [Query complexity] \label{thm:complexity:without}
    For connected weighted graph reconstruction using distance and edge-weight queries, there is a randomized algorithm with query complexity $\tilde{O}\pth{ n ^ {3/2} D^{3W_{\textrm{max}} }}$.
\end{theorem}

\section{Deferred Proofs from \cref{sec:without}}
\subsection{Correctness}
\begin{lemma}[\cref{lem:2.4} restated]

For NT-RS,    $G \subseteq \bigcup_{a \in A} G[D_a] $ where $G[D_a] $ is the subgraph of $G$ induced by vertices in $D_a $.
\end{lemma}
\begin{proof}
    $\forall uv = e \in E$, we want to show that $\exists a \in A$, such that $u, v \in D_a$. W.l.o.g., $d(A, u) \leq d(A, v)$.  $a$ is a vertex such that $d(a, u) = d(A, u)$. If $d(a, u) \leq W_{\textrm{max}}$, then $d(a, v) \leq d(a, u) + d(u, v) \leq 2W_{\textrm{max}} \Rightarrow v \in \bar{B}(a, 2W_{\textrm{max}})$. 
    
    Otherwise, $d(a, u) > W_{\textrm{max}}$. We claim that on any shortest weighted path from $a$ to $u$, there is a vertex $b$ such that $d(a, b) \in (W_{\textrm{max}}, 2W_{\textrm{max}}]$. It is possible that $b = u$.  
$d(b, u) = d(a, u) - d(a, b) = d(A, u) - d(a, b) < d(A, u)$. By the triangle inequality, 
    \begin{gather}
        d(b, v) \leq d(b, u) + d(u, v)  \\
        \leq d(b, u) + W_{\textrm{max}}     \\
        = d(A, u) - d(a, b) + W_{\textrm{max}}   \\
        \leq   d(A, v) - d(a, b) + W_{\textrm{max}} \\
        < d(A, v)
    \end{gather}
    Therefore, $u, v \in C(b)$. Since $b \in \bar{B}(a, 2W_{\textrm{max}}), u, v \in D_a$. 
    
    To prove our claim, imagine arranging such a path on an $x$-axis so that $a$ is at the origin and all edges are on its right. Check the vertices on this path from left to right. Suppose $b$ is the first vertex such that $d(a, b) > W_{\textrm{max}}$. We want to show that $b$ also satisfies the requirement $d(a, b) \leq 2W_{\textrm{max}}$ and is the vertex we want. 
    If not, $d(a, b) > 2W_{\textrm{max}}$. Suppose $b^-$ is the vertex on the left of $b$. Since $b$ is the first vertex such that $d(a, b) > W_{\textrm{max}}$, $d(a, b^-) \leq W_{\textrm{max}}$. $d(b^-, b) = d(a, b) - d(a, b^-) > W_{\textrm{max}}$. This contracts the maximum weight. 

\end{proof}
\subsection{Query Complexity}
To prove the query complexity of NT-R, we first recall the following lemma in \cite{kannan_graph_2018}. 
\begin{lemma}[Rephrasing of Lemma 2.3 in \cite{kannan_graph_2018}]
Let $k$ in ESTIMATED-CENTERS be some well-chosen constant. With probability at least $1/4$, the algorithm ESTIMATED-CENTERS (\cref{alg:centers}) terminates after \\ $O\left(  sn \log^2{n} \log\log{n} \right)$  queries, and outputs a set $A \subseteq C$, such that $|A| \leq 12s \log n$ and $\cardin{C(w)} \leq 6n/s$ for every $w \in V$ where $C(w) = \{v \in V : d(w, v) < d(A, v)\}$.

\end{lemma}
\begin{theorem} [\cref{thm:complexity:without} restated]
    For connected weighted graph reconstruction using distance and edge-weight queries, there is a randomized algorithm with query complexity $\tilde{O}\pth{ n ^ {3/2} D^{3W_{\textrm{max}} }}$.
\end{theorem}    
\begin{proof}
    We first analyze the query complexity of NT-RS.
$\forall v \in \bar{B}(a, 2W_{\textrm{max}}),   \\ d(a, v) \leq 2W_{\textrm{max}}$. Since edge weights are at least 1, the number of edges on a shortest weighted path from $a$ to $v$ is at most $2W_{\textrm{max}}$. Since the graph has maximum degree $D$, 

\begin{gather}
\cardin{\bar{B}(a, 2W_{\textrm{max}})} \leq \sum_{i = 0}^{2W_{\textrm{max}}} D^ i = \frac{D^{2W_{\textrm{max}} + 1} - 1}{D - 1} \coloneq b 
\end{gather}
By the lemma above, with probability at least $1/4$, $    \cardin{D_a} \leq b + \frac{6n}{s} b$ and the nested for loops in \cref{alg:reconstruct-sub} has query complexity 
\begin{gather}
    \sum_{a \in A} (\cardin{\bar{B}(a, 2W_{\textrm{max}})} n + \cardin{D_a} ^ 2 )   
    \leq 12 s \log{n} \pth{ b n + \pth{b + \frac{6n}{s} b} ^ 2}
\end{gather}
ESTIMATED-CENTERS has query complexity $O( sn \log^2{n} \log\log{n})$. When $s = \sqrt{bn}$, the overall complexity is $O\pth{ n ^ {3/2}  b ^ {1/2} \log{n} \pth{37 b + \log{n} \log\log{n}}}$ with probability at least $1/4$.
Assuming $D \gg 1$, we can further simplify $b = D^{2W_{\textrm{max}} }$. 
By the expectation of a geometric random variable, the expected query complexity of \cref{alg:reconstruct} is $\tilde{O}\pth{ n ^ {3/2} b ^ {3/2}} = \tilde{O}\pth{ n ^ {3/2} D^{3W_{\textrm{max}} }}$.
 The expectation is wrt the randomized algorithm. 
 \end{proof}

\end{appendices}

\end{document}